\documentclass[runningheads]{llncs}

\usepackage{chngcntr}
\usepackage{graphicx}
\usepackage{graphics}
\usepackage{amscd}
\usepackage{amsmath}
\usepackage{amsfonts}
\usepackage{amssymb}
\usepackage{dblfloatfix}
\usepackage[english]{babel}
\usepackage{xcolor}

\newcommand{\F}{\mathbb{F}}
\newcommand{\C}{{\cal C}}

\begin{document}

\title{Dihedral codes with prescribed minimum distance} 

\author{Martino Borello\inst{1}
\and
Abdelillah Jamous \inst{2}
}

\authorrunning{Martino Borello and Abdelillah Jamous.}
 
\institute{
Universit\'e Paris 8, Laboratoire de G\'eom\'etrie, Analyse et Applications,  LAGA, Universit\'e Sorbonne Paris Nord, CNRS, UMR 7539,  F-93430, Villetaneuse, France
 \and
 Faculty of Mathematics, University of Sciences and Technology Houari Boumediene, Algiers, Algeria
}

\maketitle              

\begin{abstract}
Dihedral codes, particular cases of quasi-cyclic codes, have a nice algebraic structure which allows to store them efficiently. In this paper, we investigate it and prove some lower bounds on their dimension and mi\-nimum distance, in analogy with the theory of BCH codes. This allows us to construct dihedral codes with prescribed minimum distance. In the binary case, we present some examples of optimal dihedral codes obtained by this construction.

\end{abstract}

\section{Introduction}

Block codes were invented in the forties to correct errors in the
communication through noisy channels (see \cite{Huffman} for more
details), and they are used nowadays in different areas of information security. Originally, they were thought just as subsets of
(code)words of $n$ letters chosen in an alphabet $K$, which are 
far enough apart from each other with respect to the Hamming distance. However, they usually need to have more algebraic structure to be stored efficiently. By considering \emph{linear codes} of length $n$ over a finite field $K$, that is subspaces of the vector space $K^n$, we have a compact description given, for example, by the \emph{parity check matrix}, which is a 
matrix $H$ such that $c\in \C$ if and only if $cH=0$.
Such a description reduces exponentially the size of the data to be
stored with respect to general block codes. However, this reduction reveals to be insufficient in the context of code-based cryptography (\cite{McEliece,Niederreiter} and many others), where the public key is related to the parity check matrix of a code of large length and
dimension. The size
of the public key constitutes one of the main practical disadvantages in the use of code-based cryptography and many efforts have been done to reduce it by preserving the security of the system. One option may be to use codes with symmetries, like cyclic or quasi-cyclic codes  (see for example \cite{BCGO}). However, since decoding of general quasi-cyclic codes is difficult, the algebraic structure that one needs to add may also reveal to be a weakness of the system (see for example \cite{FOPT}).

A natural generalisation of cyclic codes is given by the family of group codes: a linear code $\C$ is called a $G$-\emph{code} (or
a group code) if $\C$ is a right (or left) ideal in the group algebra $KG
= \{ a=\sum_{g\in G}a_gg \mid a_g \in G\}$ where  $G$ is a finite
group. Reed Muller
codes over prime fields $\F_p$  are group codes for an elementary
abelian $p$-group $G$ \cite{Berman,Charpin}, and there are many
other remarkable optimal codes which have been detected
as group codes \cite{Bernhardt,Conway,Felde,McLH}. If $G$ is cyclic, then all right (or left) ideals of $KG$ afford only one
check equation (and then only a few data have to be stored). In the case $G$ is a general
finite group there are only particular right (or left) ideals which satisfy
this property, called \emph{checkable} codes \cite{JLLX10}. In \cite{checkable} it is proved that such codes are the duals of principal ideals and group algebras $KG$ for which all right (or left)
ideals are checkable (or equivalently principal), called \emph{code-checkable group algebras}, are characterised: $KG$ is a
code-checkable group algebra if and only if $G$ is $p$-nilpotent
with a cyclic Sylow $p$-subgroup, where $p$ is the characteristic of $K$. This is a consequence of an early result by Passman (\cite[Theorem 4.1]{Passman}). Checkable codes are asymptotically good \cite{BM,BWasym} and many optimal codes are checkable \cite[Remark 2.9]{checkable}. This seems to suggest that the family of checkable codes is worth further investigation. In particular, it is desirable to prove some bounds on the dimension and  minimum distance for checkable codes and to introduce families of checkable or principal codes with prescribed minimum distance (in analogy with BCH codes). 

To our knowledge, there are very few results concerning the parameters of group codes, both for general and particular groups. In \cite{Elia}, an algorithm for computing the dimension of general group codes is given. In a very recent paper \cite{Claro}, several relations and bounds for the dimension of principal ideals in group algebras are determined by analysing minimal polynomials of regular representations. The concatenated structure of dihedral codes is investigated in \cite{Cao}. However, we are not aware of results which allow to construct group codes with a prescribed minimum distance or explicit lower bounds on both dimension and minimum distance, even in the easiest case of dihedral codes. This paper wants to be a first contribution in this direction. In \S\ref{SectionQC} we will recall some results of the theory of quasi-cyclic codes. In \S\ref{SectionDC} we will recall the definition of dihedral codes, present some results about their algebraic structure, make some remarks about the dual codes, prove a BCH bound for principal dihedral codes, propose a definition of principal BCH-dihedral codes, consider the particular case of binary dihedral codes and give some construction of optimal codes. Finally, in \S\ref{SectionOP} we will present some open problems. In particular, an efficient decoding algorithm would be a necessary prerequisite for applications in cryptography.

\section{Quasi-cyclic codes}\label{SectionQC}
We recall in this section some definitions and known results about quasi-cyclic codes. As we will see in the next section, dihedral codes, as all group codes, form a subfamily of quasi-cyclic codes. 

Let $q$ be a power of a prime and $\F_q$ the finite field with $q$ elements. 
Let $ n \in \mathbb{N}.$ The symmetric group $S_n$ acts on the vector space $\F_q^n$ as follows:
$$v^\sigma:=(v_{\sigma^{-1}(1)},v_{\sigma^{-1}(2)},\ldots,v_{\sigma^{-1}(n)})$$
for $v:=(v_1,v_2,\ldots,v_n)\in \F_q^n$ and $\sigma\in S_n$. For a linear code $\mathcal{C}\subseteq \F_q^n$, the set of permutations such that $\mathcal{C}^\sigma:=\{c^\sigma \mid c\in \mathcal{C}\}$ is equal to $\mathcal{C}$ is a group which is called the \emph{permutation automorphism group} of $\mathcal{C}$ and which is denoted by ${\rm PAut}(\mathcal{C})$.

In this context, a remarkable transformation is the so-called \emph{shift map}, that is
$$T_n:\mathbb{F}_{q}^{n} \to \mathbb{F}_{q}^{n} \qquad c\mapsto c^{(1 \ \ldots \ n)}=(c_n,c_1,\ldots,c_{n-1}).$$
Linear codes which are invariant under the shift or its power are the so-called quasi-cyclic codes.

\begin{definition}
Let $\C\subseteq \F_q^n$ be a linear code. Suppose that $n=\ell m$, for some positive integers $\ell$ and $m$. The code $\C$ is \emph{quasi-cyclic} of index $\ell$ if $ T_n^{\ell}(\C)= \C $, that is if $$(1 \  \ldots \ n)^\ell=\prod_{j=1}^\ell(j \ \ell+j \ 2\ell+j \ \ldots \ (m-1)\ell+j)\in {\rm PAut}(\C).$$ If $ \ell=1, $ the code $\C$ is called  \emph{cyclic}.
\end{definition}

Let $R:=\F_q[x]/(x^m-1)$. We may relabel the coordinates and consider the bijective $\F_q$-linear map 
\begin{equation}\label{eq:varphi}\varphi:\F_q^n=(\mathbb{F}_{q}^{\ell})^{m}\to R^{\ell}\end{equation} 
$$(c_{11},\ldots,c_{1\ell},\ldots,c_{m1},\ldots,c_{m\ell})\mapsto (c_{11} +\cdots+ c_{m1}x^{m-1},\ldots, c_{1\ell} +\cdots+ c_{m\ell}x^{m-1}).$$
The image of a quasi-cyclic code in $R^\ell$ is an $R$-submodule. Actually, the multiplication by $x$ corresponds to the $\ell$-th power of the shift.

\begin{remark}\label{Remark1}
There is a one-to-one correspondence between the $R$-submodules of $ R^{\ell} $ and left ideals of $ {\rm Mat}_{\ell}(R)$ (which is isomorphic, as a ring, to ${\rm Mat}_{\ell}(\F_q)[x]/(x^m-1)$). This is a particular case of the Morita equivalence for modules \cite{Morita}. The explicit one-to-one map is given as follows: to any $R$-submodule $N$ of $R^\ell$ we associate the left ideal $\mathcal{I}_N$ of $ {\rm Mat}_{\ell}(R)$ composed by matrices whose rows are elements in $N$. As already observed in \cite{Barbier}, since $R$ is a commutative principal ideal ring, every $R$-submodule $N$ of $R^\ell$ has at most $\ell$ generators, so that the left ideal $\mathcal{I}_N$ is principal (it suffices to consider the matrix whose rows are the generators and eventually some zeros). So there exists a generator of $\mathcal{I}_N$ which can be seen a polynomial in ${\rm Mat}_{\ell}(\F_q)[x]/(x^m-1)$.
\end{remark}

Let $\ell$ be a positive integer, and $\alpha \in \F_{q^{\ell}}$ be a primitive element of $\F_{q^{\ell}}/\F_q$. Recall that $\{1,\alpha,\ldots,\alpha^{\ell-1}\}$ is an $\F_q$-base of the vector space $ \F_{q^{\ell}}$. The \emph{folding} is the $\F_{q}$-linear map 
 
$$\begin{array}{c}\phi:\F_{q}^{\ell}\to \F_{q^{\ell}}=\F_{q}[\alpha]\\
(a_{1},\ldots,a_{\ell})\mapsto a_{1}+a_{2}\alpha+\cdots+a_{\ell}\alpha^{\ell-1}.
\end{array}$$

\begin{definition} 
 Let $\C\subseteq \F_q^n=(\F_q^\ell)^m$ be a linear code. The \emph{folded code} of $\C$ is $\C'=\phi^m(\C)\subseteq (\F_{q^\ell})^m$. In this case, $\C$ is the \emph{unfolded code} of $\mathcal{C}'$.
\end{definition}

\begin{remark}
Note that the folded code $\C'$ of a linear code $\C$ is an $\F_q$-linear code. Moreover, $\C$ is quasi-cyclic if and only if $\C'$ is invariant under the shift $T_m$.
\end{remark}

In the next section we will use the above equivalence and the following definition many times. 

\begin{definition}
An $\F_q$-linear code $\C\subseteq (\F_{q^\ell})^m$ which is invariant under the shift $T_m$ is called \emph{$\F_q$-linear cyclic code}.
\end{definition}

Barbier \emph{et al.} define in \cite{Barbier} the analogue of BCH codes in the quasi-cyclic case. They call them \emph{quasi-BCH} codes. In \cite{GOS}, the algebraic structure of $\F_q$-linear cyclic codes over $\F_{q^\ell}$ is studied. In next section we will explore the same concepts in the context of dihedral codes.






\section{Dihedral codes}\label{SectionDC}

Let $m\geq 3$ be an integer and $$D_{2m}:=\langle \alpha,\beta \mid \alpha^{m}=1, \beta^{2}=1, \beta\alpha=\alpha^{m-1}\beta \rangle, $$
be the \emph{dihedral group} of order $2m$. The \emph{group algebra} $\F_qD_{2m}$ is the set  $$\F_qD_{2m}:=\left\{\left.\sum_{\gamma \in D_{2m}}a_{\gamma}\gamma \ \right| \  a_\gamma\in \F_q\right\},$$ which is vector space over $\F_q$ with canonical basis $\{\gamma\}_{\gamma\in D_{2m}}$. The operations of sum and multiplication by scalars are defined in the following natural way: for any $a_{\gamma}, b_{\gamma} \in \F_q$ and $c \in \F_q$ 
$$\sum_{\gamma \in D_{2m}}a_{\gamma}\gamma + \sum_{\gamma \in D_{2m}}b_{\gamma}\gamma = \sum_{\gamma \in D_{2m}}(a_{\gamma}+b_{\gamma})\gamma,$$
$$ c\cdot \left(\sum_{\gamma \in D_{2m}}a_{\gamma}\gamma\right)= \sum\limits_{\gamma \in D_{2m}}ca_{\gamma}\gamma.$$
Moreover, $\F_qD_{2m}$ is an algebra with the product
$$\left(\sum_{\gamma \in D_{2m}}a_{\gamma}\gamma\right)\bullet\left(\sum_{\gamma \in D_{2m}}b_{\gamma}\gamma\right)= \sum_{\gamma \in D_{2m}}\left(\sum_{\mu\nu=\gamma}a_{\mu}b_{\nu}\right)\gamma.$$

\begin{definition}
A \emph{dihedral code}, or a $D_{2m}$-code, is a left ideal of $\F_qD_{2m}$.
\end{definition}

As observed in \cite{BWquasi}, a linear code of length $2m$ can be seen as a $D_{2m}$-code if and only if its automorphism group contains a subgroup isomorphic to $D_{2m}$ all of whose nontrivial elements act fixed point free on the coordinates $\{1,\ldots,2m\}$. In particular, if we consider the ordering
\begin{equation}\label{eq:ordering}
D_{2m}=\{\underbrace{1}_{b_1},\underbrace{\beta}_{b_2},\underbrace{\alpha}_{b_3},\underbrace{\alpha\beta}_{b_4},\underbrace{\alpha^2}_{b_5},\underbrace{\alpha^2\beta}_{b_6},\ldots,\underbrace{\alpha^{m-1}}_{b_{2m-1}},\underbrace{\alpha^{m-1}\beta}_{b_{2m}}\},    
\end{equation}
and the $\F_q$-linear isomorphism between $\F_q^{2m}$ and $\F_qD_{2m}$ given by $e_i\mapsto b_i$ (where $\{e_i\}$ is the canonical basis of $\F_q^{2m}$), a linear code $\C\subseteq\F_q^{2m}$ is a $D_{2m}$-code if and only if
$$\alpha':=(1 \ 3 \ 5 \ \ldots \ 2m-1)(2 \ 4 \ 6 \ \ldots \ 2m)$$
and
$$\beta':=(1 \ 2)(3 \ 2m)(4 \ 2m-1)(5 \ 2m-2)\cdots(m+1 \ m+2)$$
are in ${\rm PAut}(\C)$. These elements correspond to the permutation representation of the left multiplication by $\alpha$ and by $\beta$ respectively in $\F_qD_{2m}$.
In particular, since $\alpha'=(1 \ \ldots \ 2m)^2$, a dihedral code is a quasi-cyclic code of index $2$.

From now on, we will always consider the ordering \eqref{eq:ordering} fixed and we will identify $\F_q^{2m}$ and $\F_qD_{2m}$.

\subsection{Algebraic structure}

Let $\C$ be a $D_{2m}$-code over $\F_q$. As we observed above, since $\C$ is a quasi-cyclic codes of index $2$, $\C$ is a free left module of rank $2$ over $R:=\F_q[x]/(x^m-1)$, which is a commutative principal ideal ring. As we have already seen in Remark \ref{Remark1}, this means that $\C$ has at most two generators as a module over $R$. These are also two generators of $\C$ viewed as an ideal in $\F_qD_{2m}$. We have one generator of $\C$ as an ideal in ${\rm Mat}_2(\F_q)[x]/(x^m-1)$, given by the polynomial with coefficients in the ring of matrices with first row given by the first generator and second row given by the second one.
However, it may happen that $\C$ is not principal as an ideal in $\F_qD_{2m}$.

\begin{remark}\label{Remark3}
As observed in \cite{checkable}, an early result by Passman (\cite[Theorem 4.1]{Passman}) gives us that all $D_{2m}$-codes over a field $\F_q$ of characteristic $p$ if and only if $D_{2m}$ is $p$-nilpotent with a cyclic Sylow $p$-subgroup (we recall that a group $G$ is $p$-nilpotent if it admits a normal subgroup $N$ of order coprime with $p$ and such that $G/N$ is a $p$-group). This is the case if and only if $p$ does not divide $m$. So
\begin{itemize}
    \item if $(m,q)=1$, all $D_{2m}$-codes over $\F_q$ are principal;
    \item otherwise, a $D_{2m}$-code over $\F_q$ is either principal or the sum of two principal ideals.
\end{itemize}
\end{remark}

We will study then the algebraic structure of principal left ideals in $\F_qD_{2m}$, that is principal dihedral codes. Via the map $\varphi$ defined as in \eqref{eq:varphi}, we can consider $\varphi(\C)$ inside $R^2$. The automorphism $\alpha'$ corresponds to the multiplication by $x$ in $R^2$, whereas the automorphism $\beta'$ acts on $R^2$ as follows: for $(a(x),b(x))\in R^2$,
$$(a(x),b(x))^{\beta'}=(b(x^{m-1}),a(x^{m-1})).$$
So, $\C$ is a $D_{2m}$-code if and only if $\varphi(\C)$ is an $R$-submodule of $R^2$ invariant under the action of $\beta'$, that is such that $(b(x^{m-1}),a(x^{m-1}))\in \varphi(\C)$ for all $(a(x),b(x))\in \varphi(\C)$.

If $\C$ is principal, then $\varphi(\C)$ is an $R$-submodule of $R^2$ generated, as a module, by
$$(a(x),b(x)) \ \ \text{ and } \ \ (b(x^{m-1}),a(x^{m-1})).$$

\begin{remark}
We have already mentioned the Morita correspondence between $R$-sub\-modules and left ideals in ${\rm Mat}_2(R)\cong{\rm Mat}_2(\F_q)[x]/(x^m-1)$. In this case, the left ideal $I_\C\subseteq {\rm Mat}_2(\F_q)[x]/(x^m-1)$ associated to $\C$ is the principal ideal 
$$I_\C=\left\langle\begin{array}{c}  \begin{pmatrix} a_{0} & b_{0} \\ b_{0} & a_{0} \end{pmatrix}+ \begin{pmatrix} a_{1} & b_{1} \\ b_{m-1} & a_{m-1} \end{pmatrix}x+\cdots+\begin{pmatrix} a_{m-1} & b_{m-1} \\ b_{1} & a_{1} \end{pmatrix}x^{m-1} \end{array}\right\rangle,$$
where $a(x):=a_0+a_1x+\ldots+a_{m-1}x^{m-1}$ and $b(x):=b_0+b_1x+\ldots+b_{m-1}x^{m-1}$.
\end{remark}

Considering the folding $(\F_q)^2\to \F_{q^2}=\F_q[\alpha]$, we can see the two polynomials $a(x),b(x)$ as a unique polynomial over $\F_{q^2}$ as
$$p(x):=(a_0+b_0\alpha)+(a_1+b_1\alpha)x+\ldots+(a_{m-1}+b_{m-1}\alpha)x^{m-1}$$
so that a principal dihedral code can be seen as the sum of the two $\F_q$-linear cyclic codes over $\F_{q^2}$, that is the one generated by $p(x)$ and the one generated by $\overline{p}(x^{m-1})$, where 
$$\overline{p}(x):=(b_0+a_0\alpha)+(b_1+a_1\alpha)x+\ldots+(b_{m-1}+a_{m-1}\alpha)x^{m-1}.$$

The $\F_q$-linear map $\tau:=a+b\alpha\mapsto \overline{\tau}:=b+a\alpha$ can be expressed by the following linearised polynomial:
$$\tau\mapsto L(\tau):=\left(\frac{1-\alpha^2}{\alpha^q-\alpha}\right) \tau^q+\left(\frac{\alpha^{q+1}-1}{\alpha^q-\alpha}\right)\tau.$$
so that, if 
$$p(x):=\tau_0+\tau_1x+\ldots+\tau_{m-1}x^{m-1},$$
we have
$$\overline{p}(x)=\overline{\tau_0}+\overline{\tau_1}x+\ldots+\overline{\tau_{m-1}}x^{m-1}=$$
$$\left(\frac{1-\alpha^2}{\alpha^q-\alpha}\right) p(x^{1/q})^q+\left(\frac{\alpha^{q+1}-1}{\alpha^q-\alpha}\right)p(x).$$

\begin{definition}
For a polynomial $r(x)\in \F_{q^2}[x]/(x^m-1)$, we denote by $\langle r(x)\rangle_{\F_q}$ the unfolded $\F_q$-linear cyclic code generated by $r(x)$, i.e. the unfolded of $$\{t(x)r(x) \in \F_{q^2}[x]/(x^m-1) \mid t(x)\in \F_{q}[x]\}.$$ 
\end{definition}

We can resume all the discussion in the following.
\begin{theorem}
Let $\F_{q^2}=\F_q[\alpha]$ and $\mathcal{C}$ be a principal $D_{2m}$-code over $\F_q$. It exists $p(x)\in \F_{q^2}[x]/(x^m-1)$ such that 
$$\C=\langle p(x)\rangle_{\F_q}+\langle \overline{p}(x^{m-1})\rangle_{\F_q},$$
where
$$\overline{p}(x^{m-1})=\left(\frac{1-\alpha^2}{\alpha^q-\alpha}\right) p(x^{(m-1)/q})^q+\left(\frac{\alpha^{q+1}-1}{\alpha^q-\alpha}\right)p(x^{m-1})\in \F_{q^2}[x]/(x^m-1).$$
\end{theorem}

In particular, as we have already observed in Remark \ref{Remark3}, all $D_{2m}$-codes over $\F_q$ are principal if $(m,q)=1$ and they are a sum of at most two principal $D_{2m}$-codes otherwise.

\begin{definition}
We call the polynomial $p(x)$ a \emph{generator} of the principal dihedral code. 
\end{definition}

\begin{corollary}\label{corollary-dimension}
Let $\C$ be a principal $D_{2m}$-code over $\F_q$ generated by $p(x)$. Then
$$\dim_{\F_q} \C \geq \max\{m-\deg p(x),m-\deg \overline{p}(x^{m-1})\}.$$
\end{corollary}

\begin{proof}
This follows from the fact that the vectors in $\F_q^{2m}$ corresponding to the polynomials
$$\{p(x),xp(x),\ldots,x^{m-\deg p(x)-1}p(x)\}$$
are linearly independent, and the same holds for the ones corresponding to 
$$\{\overline{p}(x^{m-1}),x\overline{p}(x^{m-1}),\ldots,x^{m-\deg \overline{p}(x^{m-1})-1}\overline{p}(x^{m-1})\}.$$
\end{proof}

\begin{remark}\label{RemarkR}
For calculations, it may be interesting to have integer exponents. In case $(m,q)=1$, we can take $m'$ to be the inverse of $m$ modulo $q$, so that $m'm-1$ is divisible by $q$. Let $r:=(m'm-1)/q$. Then 
$$\overline{p}(x^{m-1})=\left(\frac{1-\alpha^2}{\alpha^q-\alpha}\right) p(x^r)^q+\left(\frac{\alpha^{q+1}-1}{\alpha^q-\alpha}\right)p(x^{m-1}).$$
\end{remark}

\subsection{Dual code}

In analogy with the theory of cyclic and quasi-cyclic codes, it it interesting to investigate the dual codes of dihedral codes, which are still dihedral. 

\begin{proposition}
The dual code $\C^\perp$ of a dihedral code $\C$ is a dihedral code.
\end{proposition}
\begin{proof}
This follows trivially from the fact that ${\rm PAut}(\C^\perp)={\rm PAut}(\C)$.
\end{proof}

The dual of a principal dihedral code is not necessarily principal. But if $(m,q)=1$, as we mentioned already, all dihedral codes are principal. So it makes sense to investigate the relation between the generator of a code and a generator of its dual.

Let $p(x)$ and $q(x)$ be two polynomial in $\F_{q^2}[x]/(x^m-1)$ and let $v$ and $w$ the two vectors in $\F_q^{2m}$ corresponding to $p(x)$ and $q(x)$ respectively. We may define
$$\begin{array}{ccl}\ast : \F_{q^{2}}[x]/(x^m-1) \times \F_{q^{2}}[x]/(x^m-1) &\to& \F_q\\
(p(x),q(x))&\mapsto& p(x) \ast q(x):=\langle v,w\rangle\end{array}$$

\begin{proposition}\label{prop:dual}
Let $(m,q)=1$. If $\C$ is a principal $D_{2m}$-code generated by $p(x)$ and  $\C^{\perp}$ is a principal $D_{2m}$-code generated by $q(x)$, then
$$ p(x) \ast q(x)=0, \ p(x) \ast \overline{q}(x^{m-1})=0,$$ $$\overline{p}(x^{m-1}) \ast q(x)=0, \ \overline{p}(x^{m-1}) \ast \overline{q}(x^{m-1})=0.$$
The same holds with all the shift of $p(x)$ and $\overline{p}(x^{m-1})$.
\end{proposition}

\begin{proof}
This is clear from the definition of $\ast$.
\end{proof}

\begin{remark}
At least two questions stand open in this context: the conditions in Proposition \ref{prop:dual} are only necessary. It would be very interesting to find sufficient conditions for a polynomial $q(x)$ to be a generator of the dual. We may add the orthogonality with all the shift of $p(x)$ and $\overline{p}(x^{m-1})$, but this would still be not enough. A polynomial $q(x)$ satisfying all these relations would generate a subcode of $\C^\perp$, but not necessarily the whole dual. In fact, there is an argument on the dimension missing. Secondly, it would be nice to give some relations with the usual product of polynomials (as in the cyclic codes case) and not with the $\ast$ product.
\end{remark}

For dihedral codes over fields of characteristic $2$, a nice relation holds.

\begin{proposition} If $q$ is a power of $2$, then $\langle\overline{p}(x^{m-1})\rangle_{\F_q}\subseteq \langle p(x)\rangle_{\F_q}^\perp$. In particular, the code generated by $p(x)$ is contained in $\langle p(x)\rangle_{\F_q}+\langle p(x)\rangle_{\F_q}^\perp$.
\end{proposition}

\begin{proof}
Recall that if $p(x)$ corresponds to the vector 
$$v=(a_0,b_0,a_1,b_1,\ldots,a_{m-1},b_{m-1}),$$
then $\overline{p}(x^{m-1})$ corresponds to the vector 
$$w=(b_0,a_0,b_{m-1},a_{m-1},\ldots,b_1,a_1),$$
so that 
$$\langle v,w\rangle=2(a_0b_0+a_1b_{m-1}+a_{m-1}b_1+\ldots)=0$$
in any field of characteristic $2$. Clearly, the same argument applies to $x^ip(x)$.
\end{proof}

\begin{remark}
In many examples, we get the equality $\langle\overline{p}(x^{m-1})\rangle_{\F_q}=\langle p(x)\rangle_{\F_q}^\perp$. However, we could not find a general property of $p(x)$ which guarantees it. Again, there is an argument on the dimension missing.
\end{remark}

\subsection{Minimum distance bounds}

Let $t$ be the order of $q^2$ modulo $m$, and let $\omega$ be a primitive $m$-th root of unity in $\F_{q^{2t}}$. If $\delta-1$ consecutive powers of $\omega$ are roots of both $p(x)$ and $\overline{p}(x^{m-1})$, then a BCH bound can be proved for the code generated by $p(x)$ and $\overline{p}(x^{m-1})$.

\begin{theorem}[BCH bound for principal dihedral codes]
If $\delta-1$ con\-secutive powers of $\omega$ are roots of both $p(x)$ and $\overline{p}(x^{m-1})$, then the dihedral code $\C$ generated by $p(x)$ has minimum distance at least $\delta$.
\end{theorem}
\begin{proof}
A codeword $c(x)$ of the folded $\C\subseteq \F_{q^2}^m$ is of the form
$$c(x)=t_1(x)p(x)+t_2(x)\overline{p}(x^{m-1}),$$
for $t_1(x),t_2(x)\in \F_q[x]$. As $\delta-1$ consecutive powers of $\omega$ are roots of both $p(x)$ and $\overline{p}(x^{m-1})$, we have $c(x)= c'(x)g(x)$ where $c'(x) \in \F_{q^2}[x]$ and 
$$g(x)= {\rm lcm} \{M_{\omega^b}(x),M_{\omega^{b+1}}(x),\ldots,M_{\omega^{b+\delta-2}}(x)\},$$
where $M_{\omega^i}(x)$ is the minimal polynomial of $\omega^i$ over $\F_{q^2}$. It follows that the folded $\C$ is a subcode of the BCH code genereted by $g(x)$, which has minimum distance at least $\delta$ by the classical BCH bound. Since a nonzero coordinate in a codeword of the folded $\C$ corresponds to at least a nonzero coordinate of the unfolded codeword in $\C$, the minimum distance of $\C$ is at least $\delta$.
\end{proof}

Let us consider the case $(m,q)=1$ and let $r$ be defined as in Remark \ref{RemarkR}. For many applications, it is suitable to consider codes with a prescribed minimum distance. This can be achieved by imposing that
$\delta-1$ consecutive powers of $\omega$, say $\omega^b,\omega^{b+1},\ldots,\omega^{b+\delta-2}$, together with their inverse and their $r$-th powers, are roots of $p(x)$, which guarantees that the code generated has minimum distance at least $\delta$.

\begin{definition}
Let $(m,q)=1$. A dihedral code $\mathcal{C}\subseteq\F_q^{2m}$ is a \emph{BCH-dihedral code} of prescribed minimum distance $\delta$ if it exists an integer $b$ such that its generator is 
$$p(x)={\rm lcm}\left\{\begin{array}{c}
M_{\omega^b}(x),M_{\omega^{b+1}}(x),\ldots,M_{\omega^{b+\delta-2}}(x)\\
M_{\omega^{-b}}(x),M_{\omega^{-b-1}}(x),\ldots,M_{\omega^{-b-\delta+2}}(x)\\
M_{\omega^{br}}(x),M_{\omega^{br+r}}(x),\ldots,M_{\omega^{br+\delta r-2r}}(x)
\end{array}
\right\}$$
where $r=(m'm-1)/q$, with $m'$ being the inverse of $m$ modulo $q$, and $M_{\omega^i}(x)$ is the minimal polynomial of $\omega^i$ over $\F_{q^2}$.
\end{definition}

\begin{remark}
The definition above guarantees to have minimum distance at least $\delta$. Anyway, it may probably be improved by analysing the relations between the cyclotomic cosets of the different roots. This reveals to be simpler in the binary case, that we will consider in the next subsection.
\end{remark}

\subsection{Binary case}

Let us consider now $D_{2m}$-codes over $\F_2$, with $m\geq 3$ odd.
The binary case is parti\-cularly interesting, since $\alpha^{2+1}-1=0$. In this case $$\overline{p}(x^{m-1})=\alpha p(x^{(m-1)/2})^2,$$
so that if $Z(p)$ is the set of zeros of $p(x)$, then $Z(p)^{2/(m-1)}$ is the set of zeros of $\overline{p}(x^{m-1})$. In this case, we are considering $p(x)$ and $\overline{p}(x^{m-1})$ as polynomials in $\F_4[x]$ and not in the quotient ring.

We consider an $m$-th root of unity $\omega$ in $\F_{4^t}$, where $t$ the order of $4$ modulo $m$. The irreducible divisors of $x^m-1$ are associated to the cyclotomic cosets $C_i=\{i,4i \bmod m,4^2i\bmod m,\ldots \}$ (this is classical in the theory of cyclic codes - see for example \cite{Huffman}): actually, if $M_{\omega^i}(x)$ is the polynomial associated to $C_i$ (which is the minimal polynomial of $\omega^i$), its zeros are $Z(M_{\omega^i}(x))=\{\omega^j\mid j \in C_i\}$. 

\begin{proposition}
The following conditions are equivalent:
\begin{itemize}
    \item[{\rm a)}] $Z(M_{\omega^i})^{(m-1)/2}=Z(M_{\omega^i})$ for all $i\in \{0,\ldots,m-1\}$;
    \item[{\rm b)}] $\frac{m-1}{2}C_i=C_i$ for all $i\in \{0,\ldots,m-1\}$;
    \item[{\rm c)}] it exists an integer $s$ such that $2^{2s+1}=-1\bmod m$.
\end{itemize}
If $m$ is prime, then {\rm a)}, {\rm b)} and {\rm c)} are equivalent to 
\begin{itemize}
    \item[{\rm d)}] $s_2(m)\equiv 2\bmod 4$, where $s_2(m)$ is the order of $2$ modulo $m$.
\end{itemize}
\end{proposition}

\begin{proof}
{\rm a)}$\Leftrightarrow${\rm b)}: if $Z(M_{\omega^i})^{(m-1)/2}=Z(M_{\omega^i})$, then it exists $j\in C_i$ such that $\omega^{i(m-1)/2}=\omega^j$, which means that the class $C_i$ is sent to $C_i$ by multiplying by $\frac{m-1}{2}$. The vice versa is trivial.\\
{\rm b)}$\Rightarrow${\rm c)}: since $\frac{m-1}{2}C_1=C_1$, it exists $s$ such that $\frac{m-1}{2}=4^s\bmod m$. Then $2^{2s+1}=-1\bmod m$.\\
{\rm c)}$\Rightarrow${\rm b)}: $2^{2s+1}=-1\bmod m$ implies ($(m,2)=1$ so that $2$ is invertible) that $4^s=\frac{m-1}{2}\bmod m$. This means that for all $i\in \{0,\ldots,m-1\}$, we have $\frac{m-1}{2}i=4^si\in C_i$, which implies $\frac{m-1}{2}C_i=C_i$.\\
{\rm c)}$\Rightarrow${\rm d)}: Since $2^{4s+2}=1\bmod m$ and $2^{2s+1}=-1\bmod m$, then $s_2(m)$ divides $2(2s+1)$ and $s_2(m)$ does not divide $2s+1$. So $2$ divides $s_2(m)$. If $4$ divides $s_2(m)$, then $4$ divides $4s+2$, which is not true. So $s_2(m)\equiv 2 \bmod 4$.\\
{\rm d)}$\Rightarrow${\rm c)}: If $s_2(x)=4s+2$, then $2^{2s+1}$ is a root of $x^2-1\in \F_m[x]$, which has only two solutions. The only possible solution in this case is $-1$ (otherwise the order of $2$ would be smaller than $4s+2$). 
\end{proof}

\begin{remark}
The set of primes $\mathcal{P}:=\{m\mid s_2(m)\equiv 2\bmod 4\}=\{3,11,19,43,\ldots\}$ is infinite (its density in the set of primes is 7/24 \cite{Moree}). 
\end{remark}

\begin{theorem}\label{Theorem-binary}
If it exists an integer $s$ such that $2^{2s+1}=-1\bmod m$ $($in particular if $m$ is prime and $s_2(m)\equiv 2\bmod 4)$, then, for all integers $\delta\geq 2$ and $b\geq 0$, the binary $D_{2m}$-code generated by 
$$p(x)={\rm lcm}\{
M_{\omega^b}(x),M_{\omega^{b+1}}(x),\ldots,M_{\omega^{b+\delta-2}}(x)\}$$
is a principal BCH-dihedral code with minimum distance $d\geq \delta$ and dimension $k\geq m-\deg p(x)$.
\end{theorem}

\begin{proof}
It follows from the fact that, in this case, $p(x)$ divides $\overline{p}(x^{m-1})$: actually, all roots of $p(x)$ are roots of $\overline{p}(x^{m-1})$ (as polynomial in $\F_4[x]$) and $p(x)$ divides $x^m-1$.
\end{proof}



\begin{remark}
Theorem \ref{Theorem-binary} allows to construct binary dihedral codes with prescribed minimum distance and with a lower bound on their dimensions. With {\sc Magma} we did some calculations and we found some codes with the best-known minimum distance for their dimension (see \cite{Grassl}). For example:
\begin{itemize}
\item the $D_{22}$-code generated by
$$p(x)=x^5 + \alpha x^4 + x^3 + x^2 + \alpha^2 x + 1,$$
which is a $[22, 12, 6]$ code;
\item the $D_{66}$-code generated by
$$p(x)=x^{15} + \alpha x^{14} + x^{13} + x^{11} + x^{10} + \alpha^2 x^9 + \alpha^2 x^8 +$$ $$+\alpha x^7 + \alpha  x^6 + x^5 + x^4 + x^2 + \alpha^2 x + 1,$$
which is a $[ 66, 33, 12 ]$ code;
\item the $D_{86}$-code generated by
$$p(x)=x^{21} + \alpha x^{20} + \alpha x^{18} + \alpha x^{17} + \alpha x^{16} + x^{15} + \alpha^2x^{11} + \alpha x^{10} +$$ $$+ x^6 + \alpha^2x^5 + \alpha^2x^4 + \alpha^2x^3 + \alpha^2x + 1,$$
which is a  $[86,44,15]$ code;
\item the $D_{86}$-code generated by
$$p(x)=x^7 + x^6 + \alpha x^5 + \alpha^2 x^2 + x + 1.$$
which is a $[86,72,5]$ code.
\end{itemize}
Note that the dimension is always $2(m-\deg p(x))$.
\end{remark}

\section{Open problems}\label{SectionOP}

In the paper we defined dihedral codes with prescribed minimum distance and dimension. However, it would be interesting to prove better bounds on the dimension and to give a construction allowing to control it. In particular, an open problem is the following. \\

\noindent \textbf{Problem 1.} When does equality hold in Corollary \ref{corollary-dimension}? Can the bound be improved by adding some conditions on $p(x)$?\\

Related to that, there is also the problem of a canonical generator. Actually, in the theory of BCH codes we can read the dimension from the degree of the generator polynomial (the one of lowest degree). It does not seem to exist an analogue for dihedral codes. About dual codes, many questions stand open. The main one is about the relation between the generators of code. Another important problem, related to the use of dihedral codes in cryptography is the following.\\

\noindent \textbf{Problem 2.} Is there any efficient decoding algorithm for dihedral codes, based on the algebraic structure proved in the paper?\\

Finally, it would be interesting to extend the results to other group codes, at least in the checkable case.





\begin{thebibliography}{8}

\bibitem{Barbier} M. Barbier, C. Chabot and G. Quintin, \emph{On quasi-cyclic codes as a generalization of cyclic codes}, Finite Fields Appl., 18(5), pp. 904-919, 2012.

\bibitem{BM} L.M.J.~Bazzi and S.K.~Mitter, \emph{Some randomized code constructions from group actions}, IEEE Trans. Inform. Theory 52, pp. 3210--3219, 2006.

\bibitem{BCGO} T.P. Berger, P.L. Cayrel, P. Gaborit and A. Otmani, \emph{Reducing key length of the McEliece cryptosystem}, In International Conference
 on Cryptology in Africa, Springer, Berlin, Heidelberg, pp. 77--97, 2009.

\bibitem{Berman} S.D.~Berman, \emph{On the theory of group codes}, Kibernetika 3, pp. 31--39, 1967.

\bibitem{Bernhardt} F.~Bernhardt, P.~Landrock and  O.~Manz,
\emph{The extended Golay codes considered as ideals}, J. Comb. Theory, Series A 55, pp. 235--246, 1990.

\bibitem{checkable} M. Borello, J. de la Cruz and W. Willems, \emph{On checkable codes in group algebras}, {\tt {arXiv}: 1901.10979}, 2019.

\bibitem{BWasym} M. Borello and W. Willems, \emph{Group codes over fields are asymptotically good},  {\tt {arXiv}: 1904.10885}, 2019.

\bibitem{BWquasi} M. Borello and W. Willems, \emph{On the algebraic structure of quasi group codes},  {\tt {arXiv}: 1912.09167}, 2019.

\bibitem{Cao} Y.~Cao, Y.~Cao and F.W.~Fu, \emph{Concatenated structure of left dihedral codes}, Finite Fields and Their Applications, 38, pp. 93--115, 2016.

\bibitem{Charpin} P.~Charpin, \emph{Une g{\'e}n{\'e}ralisation de la construction de Berman des codes de Reed-Muller p-aire},
Comm. Algebra 16, pp. 2231--2246, 1988.

\bibitem{Claro} E.J.G.~Claro and H.T.~Recillas, \emph{On the dimension of ideals in group algebras, and group codes}, {\tt {arXiv}: 2002.06407}, 2020.

\bibitem{Conway} J.H.~Conway, S.J.~Lomonaco Jr and N.J.A.~Sloane,
\emph{A $[45,13]$ code with minimal distance 16}, Discrete Math. 83, pp.  213--217, 1990.

\bibitem{Elia} M. Elia and E. Gorla, \emph{Computing the dimension of ideals in group algebras, with an application to coding theory}, {\tt {arXiv}: 1403.7920}, 2019.

\bibitem{FOPT} J.-C. Faug\`ere, A. Otmani, L. Perret and J.-P. Tillich, \emph{Algebraic cryptanalysis of McEliece variants with compact keys}, in: H. Gilbert (Ed.), Advances in Cryptology EUROCRYPT 2010, in: Lecture Notes in Comput. Sci., vol. 6110, Springer, Berlin, Heidelberg,
pp. 279--298, 2010.

\bibitem{Felde} A.~vom Felde, \emph{A new presentation of Cheng-Sloane's
$[32,17,8]$-code}, Arch. Math. 60, pp. 508--511, 1993.

\bibitem{GOS} C.~G\"uneri, F.~\"Ozdemir and P. Sol\'e, \emph{On the additive cyclic structure of quasi-cyclic codes}, Discrete Mathematics, 341(10): pp. 2735--2741, 2018.

\bibitem{Grassl} M. Grassl, Codetables,  http://www.codetables.de/.

\bibitem{Huffman} W. Huffman and V. Pless, \emph{Fundamentals of Error-Correcting Codes}, Cambridge, U.K. Cambridge Univ. Press, 2003.

\bibitem{JLLX10} S.~Jitman, S.~Ling, H.~Liu and X.~Xie, \emph{Checkable codes from group rings},  {\tt {arXiv}: 1012.5498v1}, 2010.

\bibitem{McEliece} R.~J.~McEliece, \emph{A public-key cryptosystem based on algebraic coding theory}, DSN Progress Report 42--44, pp. 114--116, 1978.

\bibitem{McLH} I.~McLoughlin and T. Hurley, \emph{A group ring construction of the extended binary Golay code}, IEEE Trans. Inform. Theory 54, pp. 4381--4383, 2008.

\bibitem{Moree} P. Moree, \emph{On the divisors of $a^k+ b^k$}. Acta Arithmetica, 80(3), pp. 197--212, 1997.

\bibitem{Morita} K. Morita, \emph{Duality for modules and its applications to the theory of rings with minimum condition}, Science Reports of the Tokyo Kyoiku Daigaku, Section A, 6(150), pp. 83-142, 1958.

\bibitem{Niederreiter} H. Niederreiter, \emph{Knapsack-type cryptosystems and
algebraic coding theory}. Problems of Control and Information
Theory. Problemy Upravlenija i Teorii Informacii 15,  pp.
159--166, 1986.

\bibitem{Passman}  D.S. Passman, \emph{Observations on group rings}, Comm. Algebra 5, pp. 1119-1162, 1977.

\end{thebibliography}
\end{document}